\documentclass[10pt]{amsart}

\usepackage{amssymb}
\usepackage{amsfonts,amsthm,amsmath}

\usepackage{graphicx,epstopdf}
\usepackage{setspace}
\usepackage{color}
\usepackage{xcolor}
\usepackage{multirow}
\usepackage{rotating}
\usepackage{float}
\usepackage{mathtools}

\def\E{{\mathbb{E}}}
\def\P{{\mathbb{P}}}

\def \I{{\mathbb{I}}}
\def \F{{\mathcal{F}}}
\def \Ij{ {\mathbb{I}}_{\{ t_j < \theta \}}  }

\usepackage[english]{babel}
\usepackage[left=3cm,right=3cm,top=2.5cm,bottom=3cm]{geometry}

\usepackage{amsmath,amssymb,amsthm,bbm,color,graphics,version}
\usepackage{mathrsfs}

\newcommand{\bearno}{\begin{eqnarray*}}
\newcommand{\enarno}{\end{eqnarray*}}

 \def\Xi{X^{\infty}}

\def\P{{\mathbb P}}   
\def\E{{\mathbb E}}

\newtheorem{Thm}{Theorem}

\newtheorem{Lemma}{Lemma}

\title{Pricing Time-Capped American Options Using Least Squares Monte Carlo Method}

\author{Pawe\l\ St\c{e}pniak}
\address{Faculty of Pure and Applied Mathematics, Wroc\l aw University of Science and Technology, Wroc\l aw, Poland} \email{pawel.stepniak@pwr.edu.pl}
\author{Zbigniew Palmowski}
\address{Faculty of Pure and Applied Mathematics, Wroc\l aw University of Science and Technology, Wroc\l aw, Poland} \email{zbigniew.palmowski@pwr.edu.pl}

\thanks{
This work is partially supported by National Science Centre, Poland, under grants
No. 2021/41/B/HS4/00599.}

\date{\today}
\keywords{}

\begin{document}
\begin{abstract}
In this paper, we adopt the least squares Monte Carlo (LSMC) method to price time-capped American options. The aforementioned
cap can be an independent random variable or dependent on asset price at random time. We allow various time caps. In particular, we give an algorithm for pricing the American options capped by
the first drawdown epoch. We focus on the geometric L\'evy market.
We prove that our estimator converges to the true price as one takes the discretisation step tending to zero and the number of trajectories going to infinity.
\vspace{3mm}

\noindent {\sc Keywords.}
L\'evy process $\star$ LSMC algorithm $\star$ pricing $\star$ American option $\star$ optimal stopping $\star$ Monte Carlo $\star$ cap

\noindent {\sc Key messages.}
\begin{itemize}
\item Creating a modified least squares Monte Carlo pricing method.
\item Pricing a time-capped option.
\item Performing a numerical analysis.
\end{itemize}

\end{abstract}

\maketitle

\pagestyle{myheadings} \markboth{\sc P.\ St\c{e}pniak --- Z.\ Palmowski} {\sc Time-capped American options}

\vspace{1.8cm}

\tableofcontents

\newpage

\section{Introduction}
There has been much recent interest in mathematical finance extending
classical American options to more general options that
could be exercised early but have limited risk.
One of the most celebrated derivatives of this type are so-called capped options.
The cap could be placed either on the asset value or on the maturity time.
Moreover, the cap could be deterministic or random.
The motivation for introducing capped options
comes from limiting the liability and therefore making such a product more interesting for possible investors.
Furthermore, these derivatives are also cheaper than their vanilla counterparts.
Popularity of this type of financial instrument strongly depends though on
understanding their pricing, hedging, and optimal exercise policies.
This will to understand these features was our motivation for this paper.

The most common choice is a cap put on the asset price.
One of the simplest examples of capped options are introduced in 1991 by the Chicago Board of Options Exchange European options written on the S\&P 100 and S\&P 500 with a cap on their payoff function (see \cite{cap}).
In other words: they were terminated if the index value exceeded a given threshold. For other Papers related to this type of cap, see \cite{Zaevski} and references therein.

Our paper focuses on the time cap instead, which is put on the maturity date.
Such derivatives are terminated early if a prespecified event occurs.
Usually, the default time of some company is chosen as a time-cap; see e.g. \cite[p. 27]{BieleckiRutkowski} or \cite{Mladen}.
Sometimes the time cap is chosen to be unobservable; see \cite{gapeev1}.

Other possibility is related to independent of the asset price process time-cap with exponential or Erlang distribution.
This type of cap was considered, for example, in
\cite{carr} and was used to estimate the price of the American put option by randomising the time to maturity.
This method is called the Canadian approximation and was analysed in detail in \cite{Florin}.
The value of an 'n-Erlangian' option may be computed via an iterative method
with $n$ stages. When $n\rightarrow+\infty$, this approximation converges to the true value.

Time-cap is directly related to cancellable options, which are terminated early when
some event happens. In fact, the price of the stopped option is a sum of the cancellable option and the
discounted pay-out (under risk-neutral measure) at event time if it happens before maturity.
Usually this event is described as the first or last time when the underlying asset price
hits a specific threshold; see, e.g., \cite{1st_pass, algo} and references therein.

Similar early termination features have game options (like, e.g., Israeli options) where a seller has a right to terminate early at the cost of fixed penalty paid to the buyer; see the seminal paper of \cite{kifer} and \cite{meyer, Zaevski2, Zaevski2b}.

There are works dealing with derivatives that are very similar to the ones analysed in this paper.
Egloff, Farkas and Leippold in \cite{time_constraints} price American options with
stochastic stopping time constraints where the constraint put on the exercise rule allows one to exercise the option only when some prespecified conditions are satisfied. These conditions are expressed in terms of the states of a Markov process and can be related to a stochastic performance condition.
As the authors noted, such
performance-based constraints not only play an important role for structuring new
investment products, but also for the design of executive stock option plans with
exercise constraints based, e.g., on the out-performance of a reference index.
To some extent, our motivation for considering a time-capped option is very similar:
we want to introduce the option that runs up to either maturity or some event described above, whatever happens first.
However, the difference in pricing is crucial. In \cite{time_constraints} the buyer can exercise an option only when some predetermined condition, associated with the performance of the underlying asset, is satisfied. Therefore, it was possible to transform the constrained pricing problem into an unconstrained optimal stopping problem that corresponds to a generalised barrier option pricing and a stochastic Cauchy-Dirichlet problem.

There are also some papers that analyse the time-cap of American options.
In \cite{trabelsi} the random is a~first hitting time of a fixed barrier by
the underlying asset price. In \cite{ott} the fixed time cap is analysed for
Russian and American look-back options. Finally, there are other related works.
For example, \cite{russian, random_put} analyse Russian options terminated when the stock price hits its running maximum for the last time, and American options terminated when the stock price hits a certain pre-specified level for the last time.

Our main goal is to create a general numerical method of pricing time-capped options.  For this purpose we chose the Least Squares Monte Carlo (LSMC) method proposed by Longstaff and Schwartz in \cite{lsmc}. The main reason for this choice is its robustness to model choice.
Additionally,  it does not fail due to the so-called curse of dimensionality, in contrast to, e.g., finite-difference method or binomial trees.
LSMC algorithm is based on an approximate estimate of the expected value of holding the option conditional on the underlying asset price history using a linear combination of the basis polynomials from a chosen orthogonal space. Under the assumption that the holder of the option should exercise it as soon as it is profitable, this procedure allows one to find the stopping time for each of the simulated underlying asset price trajectories. The convergence of this method has been proved by Stentoft \cite{stentoft} and Clement, Lamberton and Protter \cite{proof}.

In order to price options with the time cap, we introduce a modified LSMC method. It allows not only to price a new class of instruments, but it is also suitable for options written on the underlying asset described by the geometric L\'evy process.
Such choice of the asset price process allows us to better reflect the market behaviour, for instance, the market crashes.

Indeed, several empirical studies show that the log prices of stocks have a heavier left tail than the normal distribution, on which the seminal Black-Scholes model is founded.
Since the seminal paper of Merton \cite{Merton} who introduced the geometric L\'evy process, many papers have used them (see
\cite{Cont, Schoutens} for an overview and references). In particular,
the following evolutions of the stock price process were considered:
the normal inverse Gaussian model (see \cite{B10}), the hyperbolic
model (see \cite{B42}), the variance gamma model
(see \cite{B80}), the CGMY model (see \cite{B24}),
the tempered stable process (see
\cite{B68, boyarchenkolevendorskii}.
The L\'evy markets appear in
pricing of American options as well; see, e.g.
\cite{Aase2010, AliliKyprianou, AsmussenAvramPistorius, Erik2, Erik1, boyarchenkolevendorskii, carr_wu, Mordecki}, and references therein.

To use the LSMC method, we adapt the method of \cite{proof} and prove the convergence of our algorithm when the discrimination time step tends to zero and the number of Monte Carlo simulations of the underlying asset process approaches infinity. Finally, we conduct a detailed numerical analysis for the different types of time caps.

Our main example that we will focus on is the cap related to the first time when the drawdown of the asset price will be bigger than a fixed level.
By the drawdown of a price process, we mean here the distance of the current value away from its maximum value it has attained to date.
The protection of the seller against large drawdowns written explicitly in the financial contract is very natural
to minimise the possible loss. Therefore, the list of papers pricing contracts with drawdown or drawup feature is quite long; see,
e.g. \cite{CZH, GZ, MA, ZP&JT2, ZP&JT, drawdownup2, PV, Sorn, Vec1,Vec2, drawdownup1, olympia}
and references therein.

This paper is organised as follows. In the next section we present the market in which we price the capped American option.
In Section \ref{sec:alg} we present the algorithm used and in Section
\ref{sec:conv} we prove its convergence to the true price of the option considered.
Finally, in Section \ref{sec:num} we present a numerical analysis concentrating on the cap determined by the drawdown event.

\section{Market setup}\label{sec:market}
Let us assume a L\'evy market in which the asset price is described by the following process:
\begin{equation}\label{S_def}
S_t = e^{X_t},
\end{equation}
where $X_t$ is a spectrally negative L\'evy process and $s = S_0$ is an initial asset price. More specifically, we choose
\begin{equation}\label{X_def}
	X_t = x + \mu t + \sigma B_t  - \sum_{k=1}^{N_t}U_k, 
\end{equation}
where $x = X_0 = \log s$ and $\sigma \geq 0$. In (\ref{X_def}) $\mu$ is a fixed drift, $B_t$ is a Brownian motion, $N_t$ is a homogeneous Poisson process with intensity $\lambda$ and $\{U_k\}_{\{k\in \mathbb{N}\}}$ is a sequence of independent identically distributed random variables with a finite second moment. We assume that $B_t$, $N_t$ and $\{U_k\}_{\{k\in \mathbb{N}\}}$ are mutually independent. We introduce a random variable $\theta$ and we extend the probability space, where process $X_t$ is constructed, to have both random objects defined on the same filtered probability space $(\Omega, \F, \{\mathcal{F}_t\}_{\{t\geq 0\}}, \P)$.
We allow $\lambda = 0$, which leads to the market described by the standard Black-Scholes model. Additionally, for simplicity, we assume that no dividend is paid to the holders of the underlying asset.

The main objective of this paper is to propose a general numerical method of pricing time-capped American options. More formally, our goal is to calculate the following value function:
\begin{equation}\label{define_V}
	V_s = \sup_{\tau \in \mathcal{T}, \tau \leq T} \E[e^{-r{\tau\wedge\theta}} G(S_{\tau\wedge\theta}) |S_0=s],
\end{equation}
where $G(\cdot)$ is the payout function, $K$ is a strike price, T is maturity date, $r$ is a constant positive risk-free interest rate, $\mathcal{T}$ is a family of stopping times and $\theta$ is the time-cap, i.e. the moment at which an event triggering termination of the option occurs.
We assume that $\theta$ is a stopping time with respect to $\{\F_t\}_{\{t \geq 0\}}$. The most common examples of the payout function are $G(S) = (K - S)^+$ for the put option and $G(S) = (S - K)^+$ for the call option. In this paper, we limit our attention to square-integrable payoff functions. Thus, the payoff function belongs to a Hilbert space. The expectation in (\ref{define_V}) is calculated with respect to the martingale measure $\P$, i.e $e^{-rt}S_t$ is a $\P$ local martingale. The choice of the jump-diffusion model as in (\ref{X_def}) results in a loss of completeness of the market; see \cite{Cont}. Therefore, the uniqueness of the equivalent martingale measure is not ensured. However, any martingale measure can be chosen and the price of the option will not be affected.

Let us define a Laplace exponent of the process $X_t$ as
\begin{equation}
	\Psi(z) = \frac{1}{t}\log\E e^{z X_t}.
\end{equation}

For the process $X_t$ defined as in (\ref{X_def}) we have
\begin{equation}\label{Psidef}
	\Psi(z) = \mu z + \frac{\sigma^2 z^2}{2} + \lambda (\eta(z) - 1),
\end{equation}
where
\begin{equation}\label{E_U}
	\eta(z) = \E e^{-z U_1}.
\end{equation}
We are only interested in this distribution of $U_1$ for which $\eta(1)$ is finite. As $e^{-rt}S_t$ is a local martingale under $\P$, then
\begin{equation}\label{Psi1}
	\Psi(1) = r.
\end{equation}
In consequence
\begin{equation}\label{martmeasure}
	\mu = r - \frac{\sigma^2}{2} +  \lambda (1 - \eta(1)).
\end{equation}

\section{Algorithm modification}\label{sec:alg}

American options give their holders the right to exercise them at any time before the maturity time T. To approximate their fair market price using numerical methods, it is necessary to discretise the possible exercise moments, that is, to choose for some $L$ a sequence $\{ t_k \}_{k=0}^L$ such that $0 = t_0 < t_1 < \ldots < t_L =T$. Then it is assumed that the option can only be exercised at time $t_k$ for any $0 \leq k \leq L$. Such discretisation transforms the American option into the so-called Bermuda option, whose price converges to the American option price for $L \to \infty$.

This section is organised as follows: first, we explain the original algorithm, which was designed for the underlying asset described by the geometric Brownian motion. Such assumption is one of the key elements of the Black-Scholes model. We can turn our exponential L\'evy process into GBM by setting $\lambda = 0$.

In subsection \ref{mod_lsmc}, we present our modified version of the algorithm, which not only allows to price time-capped options, but is also suitable for general class of spectrally negative L\'evy processes.

\subsection{Original algorithm}

The Least Squares Monte Carlo method proposed by Longstaff and Schwartz is a dynamic programming framework for pricing American options, see \cite{lsmc}. It begins by simulating the $N$ trajectories of the underlying asset price
process for some sufficiently large $N$. Each trajectory is a realisation of a geometric Brownian motion. Then, for each trajectory, the payoff at the maturity time is calculated. Then a recursive procedure starts. For each time step $t_k,\ 0 \leq k < L$, and for each trajectory, the expected value of the option payout conditional on the current underlying asset price $V_{t_{i}}$ is approximated. It allows one to determine if it is more profitable to hold the option instead of exercising it. Thus, for each trajectory, the optimal stopping time is found. Now, since the payoff function belongs to a Hilbert space, we can present the conditional expectation of the price function as a countable linear combination of some basis functions $\phi(\cdot)$ of this space, i.e.:

\begin{equation}\label{lin_comb}
	\E [e^{-r(t_{i+1}-t_i)}V_{t_{i+1}} | \F_{t_i}] = \sum_{k=0}^{\infty}\alpha_k\phi_k (S_{t_i}),
\end{equation}
for some unknown a priori coefficients $a_k$.  Usual choice of the basis assumes that
\begin{enumerate}
	\item for all $1 \leq j \leq L-1$ the sequence $(\phi_k(S_{t_j}))_{k \geq 0}$ is total in $L^2(\sigma(S_{t_j}))$,
	\item for all $0 \leq j \leq L-1$ and $M \geq 1$, if $\sum_{k=0}^{M}\eta_k\phi_k (S_{t_j}) = 0 $ almost surely then $\eta_k = 0$ for all $1 \leq k \leq M$.
\end{enumerate}
In \cite{lsmc}, the authors choose weighted Laguerre polynomials for the sake of an example. They also recommend Hermite, Legendre, Chebyshev, Gegenbauer or Jacobi polynomials as the possible alternatives. However, other basis functions are also possible, as long as they consistent with the assumptions.
Then, the approximation is calculated taking the finite number $M$ of basis functions and finding the estimators $\alpha^*_k$ of the coefficients $\alpha_k$ minimising the following error in the least squares sense:
	\begin{equation}\label{norm}
	\left\| \sum_{j=1}^{M} \alpha^*_j(t_i) \phi_j(\boldsymbol{S}_{t_i}) - e^{-r(t_{i+1}-t_i)}\boldsymbol{V}_{t_{i+1}} \right\|,
\end{equation}
where $\boldsymbol{V}_t$ and $\boldsymbol{S}_t$ denote the $N$-dimensional vectors of option and underlying asset prices at time $t$ consisting of observations from each trajectory.

Observe, that in fact this minimisation problem is equivalent to the one in linear regression in a matrix form.  
 Therefore, coefficients $a^{*}_j(t_i)$ can be calculated by solving the following equation:
\begin{equation}\label{matrix}
	{ \renewcommand\arraystretch{2.7}
		\begin{bmatrix} 
			a^{*}_1(t_i) \\ a^{*}_2(t_i) \\ \vdots \\ a^{*}_M(t_i) 
		\end{bmatrix} = 			
		\begin{bmatrix} 
			\displaystyle\sum_{k=1}^{N} \phi_1(S_{t_i}^k)\phi_1(S_{t_i}^k) && \ldots && \displaystyle\sum_{k=1}^{N}\phi_1(S_{t_i}^k)\phi_M(S_{t_i}^k)  \\ 
			
			\displaystyle\sum_{k=1}^{N} \phi_2(S_{t_i}^k)\phi_1(S_{t_i}^k) && \ldots && \displaystyle\sum_{k=1}^{N}\phi_2(S_{t_i}^k)\phi_M(S_{t_i}^k)   \\ 
			
			\vdots  && \ddots && \vdots \\
			
			\displaystyle\sum_{k=1}^{N} \phi_M(S_{t_i}^k)\phi_1(S_{t_i}^k) && \ldots && \displaystyle\sum_{k=1}^{N}\phi_M(S_{t_i}^k)\phi_M(S_{t_i}^k) 
		\end{bmatrix}^{-1}
		\begin{bmatrix} 
			\displaystyle\sum_{k=1}^{N}\text{e}^{-r\delta_{t_i}} \phi_1(S_{t_i}^k)V_{t_{i+1}}^k  \\ 
			\displaystyle\sum_{k=1}^{N}\text{e}^{-r\delta_{t_i}} \phi_2(S_{t_i}^k)V_{t_{i+1}}^k  \\ \vdots 
			\\ \displaystyle\sum_{k=1}^{N}\text{e}^{-r\delta_{t_i}} \phi_M(S_{t_i}^k)V_{t_{i+1}}^k 
		\end{bmatrix},
	}
\end{equation}
where $S_{t_{i}}^k$ is the price of the underlying asset at the moment $t_i$ for $k$-th trajectory out of $N$, $V_{t_{i+1}}^k$ is the value of the option at the moment $t_{i+1}$ for $k$-th trajectory and $\delta_{t_i} = t_{i+1}-t_i$.

In more details, the algorithm looks as follows:
\begin{enumerate}
	\item For a chosen time discretisation $0 = t_0 < t_1 < \ldots < t_L =T$ generate N trajectories of the underlying asset price process (geometric Brownian motion).
	\item Set the estimated value of the option for each trajectory at the maturity time to $V^{*k}_T = G(S_T^k), k = 1, \ldots, N$, where $S_T^k$ denotes price of the underlying asset at time $T$ for $k$-th trajectory.
	\item Find coefficients $a_j^*(t_{K-1})$ by minimizing the norm (\ref{norm})
	\item For each trajectory update the value function using the formula:
	\begin{equation}
	V^{*k}_{t_{i-1}} = \begin{cases}
		G(S^k_{t_i})\quad \text{if}\ G(S^k_{t_i}) \geq \sum_{j=0}^{M} \alpha^*_j(t_i) \phi_j(S^k_{t_i})\\
		e^{-r(t_{i+1}-t_i)}V^{*k}_{t_{i}}\quad \text{otherwise.}
	\end{cases}
	\end{equation}
	\item repeat steps (3) and (4) until $\boldsymbol{V}_0$ is reached.
	\item Calculate the price of the option by taking the average of $\boldsymbol{V}_0$ vector.
\end{enumerate}

\subsection{Modified approach}\label{mod_lsmc}
We propose a modified LSMC method suitable for pricing a wider range of financial instruments. We are interested in pricing time-capped American options. These contracts can be exercised at any time up to the maturity time T or a random time $\theta$, whichever comes first. This random time $\theta$ is the first time when a pre-determined event, which terminates the contract, happens. Such event does not have to be associated directly with the underlying asset performance and might as well be modelled by a chosen random variable independent of other stochastic processes driving the model.

Our procedure is as follows: we start the algorithm by discretising the time, i.e. choosing $\{ t_k \}_{k=0}^L$ such that $0 = t_0 < t_1 < \ldots < t_L =T$.
We restrict the values of the stopping rules only to this nest times.
Then we simulate N independent trajectories of the underlying asset prices, driven by the exponential L\'evy process. Then, for each trajectory we either simulate the random time $\theta$, if it is independent of the underlying asset price, or check if the performance of the underlying led to $\theta < T$ for that trajectory.

In the next step, the recursive procedure starts. The procedure at the terminal time $T$ is unchanged towards the standard method: for each trajectory we set $V_T = G(S_T)$.
Then, for each time step $t_j,\ 0 \leq j < T$, the procedure of estimating the expected value $V_{t_{i}}$ of continuation must take into account the stopping time $\theta$. To do so we use counterpart of equation (\ref{lin_comb}) and write:
\begin{equation}\label{ind_comb}
	\E [e^{-r(t_{i+1}-t_i)}V_{t_{i+1}} \I_{\{ t_i < \theta \}} | \F_{t_i}] = \I_{\{ t_i < \theta \}} \sum_{k=0}^{\infty}\alpha_k\phi_k (S_{t_i}),
\end{equation}
where $\I_A $ is an indicator of an event $A$.
Note that with the representation of the conditional expected value of continuation given by (\ref{ind_comb}) is the same as in (\ref{lin_comb}) until time $\theta$ is reached. Then, at all time steps from $\theta$ to T, the expected value of continuation is equal to zero due to the indicator $\I\{ t_i \neq \theta \}$. Due to this, the algorithm assures that the option can be only exercised up to the moment when it is capped.
This method does not change the procedure of estimating coefficients $\alpha_j^*$ and to do that one still needs to minimise the error (\ref{norm}).

To sum up, the modified procedure looks as follows:
\begin{enumerate}
	\item For a chosen time discretisation $0 = t_0 < t_1 < \ldots < t_L =T$ generate N trajectories of the underlying asset price process (geometric L\'evy process).
	\item for each trajectory find $t_k = \theta$ provided that $\theta < T$.
	\item Set the estimated value of the option for each trajectory at the maturity time to $V^{*k}_T = G(S_T^k), k = 1, \ldots, N$, where $S_T^k$ denotes price of the underlying asset at time $T$ for $k$-th trajectory.
	\item Find coefficients $a_j^*(t_{K-1})$ by minimizing the norm (\ref{norm}).
	\item For each trajectory update the value function using the formula:
	\begin{equation}
		V^{*k}_{t_{i-1}} = \begin{cases}
			G(S^k_{t_i})\quad \text{if}\ G(S^k_{t_i}) \geq  \I_{\{ t_i < \theta \}} \sum_{j=0}^{M} \alpha^*_j(t_i) \phi_j(S^k_{t_i})\\
			e^{-r(t_{i+1}-t_i)}V^{*k}_{t_{i}}\quad \text{otherwise.}
		\end{cases}
	\end{equation}
	\item repeat steps (3) and (4) until $\boldsymbol{V}_0$ is reached.
	\item Calculate the price of the option by taking the average of $\boldsymbol{V}_0$ vector.
\end{enumerate}

\section{Algorithm convergence}\label{sec:conv}

We prove that the modified algorithm converges to a true value.
We will use similar method as the one 
proposed by \cite{proof}.
We formulate the pricing LSMC procedure from more general perspective of the optimal stopping with a random time-cap.
More precisely,
we introduce a capped asset price via
\begin{equation}\label{stopped_S}
	\overline{S_t} = S_{t\wedge\theta}
\end{equation}
and define:
\begin{equation}
	\overline{Z}_{t_j} = e^{-r{t_j}} G(\overline{S}_{t_j}).
\end{equation}
Observe, that
\begin{equation}
	\sup_{\tau \in \mathcal{T}_L, \tau \leq T} \E \overline{Z}_\tau = \sup_{\tau \in \mathcal{T}_L, \tau \leq T} \E[e^{-r{\tau\wedge\theta}} G(S_{\tau\wedge\theta})],
\end{equation}
where $\mathcal{T}_L$ is the set of stopping times with values in moments $t_j$.

In other words, our pricing problem can re-formulated as finding
\begin{equation}
	V = \sup_{\tau \in \mathcal{T}_L, \tau \leq T} \E \overline{Z}_\tau,
\end{equation}
where $(\overline{Z}_{t_j})_{0 \leq j \leq L}$ is a sequence of square integrable random variables and $\mathcal{T}$ is a set of all possible stopping times.
Additionally, we assume in more general set-up that
\begin{equation}
	\overline{Z}_{t_j} = f(t_j, \overline{S}_{t_j})
\end{equation}
for some Borel function $f(t, s)$ that is non-increasing with respect of $t$.

We introduce a Snell envelope
\begin{equation}
	\overline{U}_{t_j} = \text{ess-}\!\!\!\!\! \sup_{\tau \in \mathcal{T}_L, \tau \leq T} \E (\overline{Z}_\tau | \F_{t_j}), \quad j = 0, \ldots, L,
\end{equation}
to formulate the dynamic programming principle
\begin{equation}
	\begin{cases}
		\overline{U}_{t_L} = \overline{Z}_{t_L} \\
		\overline{U}_{t_j} = \max (\overline{Z}_{t_j}, \E (\overline{U}_{t_{j+1}} | \F_{t_j})), \quad 0 \leq j \leq L-1.
	\end{cases}
\end{equation}
Let
\begin{equation}
	\overline{\tau}_j = \min \{ k \geq j | \overline{U}_{t_k} = \overline{Z}_{t_k} \}.
\end{equation}
Note, that $\overline{U}_{t_j} = \E [ \overline{Z}_{\overline{\tau}_j} | \mathcal{F}_{t_j}]$ and $\E \overline{U}_0 = \sup_{\tau \in \{ \overline{\tau}_0, \ldots, \overline{\tau}_L \}} \E \overline{Z}_{\tau}  = \E \overline{Z}_{\overline{\tau}_0}$. Then
above dynamic programming principle can be rewritten as follows
\begin{equation}
	\begin{cases}
		\overline{\tau}_L = t_L \\
		\overline{\tau}_j =t_j\I_{\{ \overline{Z}_{t_j} \geq \E[\overline{Z}_{\overline{\tau}_{j+1}} | \F_{t_j}] \}} + \overline{\tau}_{j+1}\I_{\{ \overline{Z}_{t_j} < \E[\overline{Z}_{\overline{\tau}_{j+1}} | \F_{t_j}] \}}, \quad 0 \leq j \leq L-1.
	\end{cases}
\end{equation}
Procedure of identifying the optimal stopping time based on $\overline{\tau}_j$
will be used in our numerical calculations. However,
to do so we have to make further adjustments.

To see that some modifications are required,  observe that for $t_j \geq \theta$ clearly {$\E[\overline{Z}_{\overline{\tau}_{j+1}} | \F_{t_j}] = \E[ f(\overline{\tau}_{j+1}, S_\theta) | \F_{t_j}] \leq \E[ f(t_j, S_\theta) | \F_{t_j}] = \E[\overline{Z}_{t_j} | \F_{t_j}] = \overline{Z}_{t_j}$},
which means that $\overline{\tau}_j = {t_j}$. However, if we apply the direct approximation in the numerical algorithm:
\begin{equation}
	\E[\overline{Z}_{t_{j+1}} | \F_{t_j}] = \E [e^{-r(t_{j+1}-t_j)}V_{t_{j+1}} | \F_{t_j}] \approx \sum_{k=0}^{M-1}\alpha^*_k\phi_k (\overline{S}_{t_j}),
\end{equation}
we might get $\E[\overline{Z}_{\overline{\tau}_{j+1}} | \F_{t_j}] > \overline{Z}_{t_j}$ as coefficients $\alpha^*_k$ are calculated based on all trajectories, also the ones not stopped up to the moment $t_j$.
For this reason, we introduce

\begin{equation}\label{tau}
	\begin{cases}
		\tau_L = t_L \\
		\tau_{j} ={t_j}\I_{ \{Z_{t_j}\Ij \geq \E[Z_{\tau_{j+1}} | \F_{t_j}]\Ij\} } + \tau_{j+1}\I_{\{ Z_{t_j}\Ij < \E[Z_{\tau_{j+1}}\Ij | \F_{t_j}] \}}, \quad 0 \leq j \leq L-1.
	\end{cases}
\end{equation}

We start from the first key lemma.
\begin{Lemma}\label{lemma1}
For all $0 \leq j \leq L$ we have
\begin{equation*}
	\overline{\tau}_j = \tau_j.
\end{equation*}
\end{Lemma}
\begin{proof}
If $t_j = t_L=T$
then $\overline{\tau}_j = T = \tau_j$.

Let us consider now the case when
$ \theta \leq t_j < T$.
Then
\begin{align}
	\overline{\tau}_j &= t_j\I_{\{ \overline{Z}_{t_j} \geq \E[\overline{Z}_{\overline{\tau}_{j+1}} | \F_{t_j}] \}} + \overline{\tau}_{j+1}\I_{\{ \overline{Z}_{t_j} < \E[\overline{Z}_{\overline{\tau}_{j+1}} | \F_{t_j}] \}}  \\
	& = t_j\I_{\{ f(t_j,\overline{S}_j) \geq \E[f(\overline{\tau}_{j+1},\overline{S}_{\overline{\tau}_{j+1}}) | \F_{t_j}] \}} + \overline{\tau}_{j+1}\I_{\{ f(t_j,\overline{S}_j) < \E[f(\overline{\tau}_{j+1},\overline{S}_{\overline{\tau}_{j+1}}) | \F_{t_j}] \}}  \nonumber \\
	& = t_j\I_{\{ f(t_j,S_\theta) \geq \E[f(\overline{\tau}_{j+1},S_\theta) | \F_{t_j}] \}} + \overline{\tau}_{j+1}\I_{\{ f(t_j,S_\theta) < \E[f(\overline{\tau}_{j+1},S_\theta) | \F_{t_j}] \}}. \nonumber
\end{align}

By the fact that $\overline{\tau}_{j+1} > t_j$ and by monotonicity of the function $f$ with respect of time, we have $f(t_j,S_\theta) \geq \E[f(\overline{\tau}_{j+1},S_\theta) | \F_{t_j}]$. In consequence, $\overline{\tau}_j = t_j$.
Therefore
\begin{align}
	\tau_j &= t_j\I_{ Z_{t_j}\Ij \geq \E[Z_{\tau_{j+1}} | \F_{t_j}]\Ij } + \tau_{j+1}\I_{ Z_{t_j}\Ij < \E[Z_{\tau_{j+1}}\Ij | \F_{t_j}] }  \\
	&= t_j\I_{ \{ 0 \geq 0\}} + \tau_{j+1}\I_{ \{ 0 < 0\}} = t_j = \overline{\tau}_j. \nonumber
\end{align}

We consider now the remaining case of
$0 \leq t_j < \theta$.
We prove this part by backward mathematical induction.
We recall that from previous considerations we know that
$\overline{\tau}_j=\tau_j$
for  $j$ such that $t_j = \theta$.
Assuming that $\overline{\tau}_{j+1}=\tau_{j+1}$, we have
\begin{align}
	\overline{\tau}_j &= j\I_{\{ \overline{Z}_{t_j} \geq \E[\overline{Z}_{\overline{\tau}_{j+1}} | \F_{t_j}] \}} + \overline{\tau}_{j+1}\I_{\{ \overline{Z}_{t_j} < \E[\overline{Z}_{\overline{\tau}_{j+1}} | \F_{t_j}] \} } \\
	& = j\I_{\{ f(j,\overline{S}_j) \geq \E[f(\overline{\tau}_{j+1},\overline{S}_{\overline{\tau}_{j+1}}) | \F_{t_j}] \}} + \overline{\tau}_{j+1}\I_{\{ f(j,\overline{S}_j) < \E[f(\overline{\tau}_{j+1},\overline{S}_{\overline{\tau}_{j+1}}) | \F_{t_j}] \} } \nonumber \\
	& = j\I_{\{ f(j,S_{t_j}) \geq \E[f(\tau_{j+1},S_{\tau_{j+1}}) | \F_{t_j}] \}} + \tau_{j+1}\I_{\{ f(j,S_{t_j}) < \E[f(\tau_{j+1},S_{\tau_{j+1}}) | \F_{t_j}] \} } \nonumber \\
	&= j\I_{\{ Z_{t_j} \geq \E[Z_{\tau_{j+1}} | \F_{t_j}] \}} + \tau_{j+1}\I_{\{ Z_{t_j} < \E[Z_{\tau_{j+1}} | \F_{t_j}] \} }  \nonumber \\
	&= j\I_{\{ Z_{t_j}\Ij \geq \E[Z_{\tau_{j+1}} | \F_{t_j}]\Ij \}} + \tau_{j+1}\I_{\{ Z_{t_j}\Ij < \E[Z_{\tau_{j+1}}\Ij | \F_{t_j}] \}} = \tau_j. \nonumber
\end{align}
This completes the proof.
\end{proof}

From Lemma \ref{lemma1} we can conclude that
\begin{equation*}
	V = \sup_{\tau \in \{ \tau_0, \ldots, \tau_L \}} \E Z_\tau,
\end{equation*}
where
\begin{equation*}
	Z_{t_j} = f(t_j, S_{t_j}).
\end{equation*}

Denoting
\begin{equation}\label{U}
	U_j = \E (Z_{\tau_j} | \F_{t_j})
\end{equation}
we have

\begin{equation}\label{EU=V}
	\E {U}_0 = \sup_{\tau \in \{ \tau_0, \ldots, \tau_L \}} \E Z_{\tau} = V.
\end{equation}
Now, given the approximation $\E [e^{-r(t_{j+1}-t_j)}V_{t_{j+1}} | \F_{t_j}] \approx \sum_{k=0}^{M-1}\alpha^*_k\phi_k (S_{t_j}) = \alpha_j^M \cdot \phi^M(S_{t_j})$ for a~chosen value of $M$ and corresponding vectors $\alpha_j^M = \{\alpha^*_0,\ldots, \alpha^*_{M-1}\}$ (with $\{\alpha^*_0,\ldots, \alpha^*_{M-1}\}$ estimated separately for each $j$ and $\phi^M = \{\phi_0,\ldots, \phi_{M-1}\}$, let us introduce stopping times $\tau_j^M$ as:

\begin{equation}\label{tauM}
	\begin{cases}
		\tau_L^M = t_L \\
		\tau_{j}^M ={t_j}\I_{\{ Z_{t_j}\Ij \geq [\alpha_j^M \cdot \phi^M(S_{t_j})]\Ij\} } + \tau_{j+1}^M\I_{ \{Z_{t_j}\Ij < [\alpha_j^M \cdot \phi^M(S_{t_j})]\Ij \}}, \quad 0 \leq j \leq L-1.
	\end{cases}
\end{equation}
Such stopping times allow us to define an approximation of the value by 
\begin{equation}
	U_0^M = \max (Z_0, \E Z_{\tau_1^M}).
\end{equation}
As $Z_0$ is deterministic, it is sufficient to find $\E Z_{\tau_1^M}$ to be able to obtain the approximation of the option price. For this purpose, Monte Carlo method is used.
Let $N$ denote the total number of simulated trajectories. Then, the approximation of stopping times $\tau_j$ for $n$-th trajectory is obtained by using only the first $M$  basis functions $\phi_k$.
In other words, $\tau_j$ are approximated by
\begin{equation}\label{tauM_N}
	\begin{cases}
		\tau_L^{M,N,n} = t_L \\
		\tau_{j}^{M,N,n} ={t_j}\I_{\{ Z^n_{t_j}\Ij \geq [\alpha_j^{M,N} \cdot \phi^M(S_{t_j}^n)]\Ij\} } + \tau_{j+1}^{M,N,n}\I_{\{ Z^n_{t_j}\Ij < [\alpha_j^{M,N} \cdot \phi^M(S_{t_j}^n)]\Ij \}}, \quad 0 \leq j \leq L-1,
	\end{cases}
\end{equation}
where $S^n_{t_j}$ is the value of $n$-th trajectory of the underlying asset process, $Z^n_{t_j} = f(t_j,S^n_{t_j})$ and $\alpha_j^{M,N}$ is a set of first $M$ estimators $\alpha^*_k$ of coefficients $\alpha_k$ in equation (\ref{lin_comb}) evaluated for $n$-th trajectory at time $t_j$. In consequence:
\begin{equation}
	U_0^{M,N} = \max (Z_0, \frac{1}{N}\sum_{n=1}^{N} Z^n_{\tau_1^{M,N,n}}).
\end{equation}
is a LSMC estimator of the value $V$.

It appears that our estimator is consistent in the following sense.
In consequence, our algorithm allows to obtain a proper approximation of the capped option price. 
\begin{Thm}\label{thm1}
Assume that 
for $0 < j < L$ we have $\P(\alpha_j \cdot \phi(S_{t_j}) = Z_{t_j}) = 0$. Then
\begin{equation}\label{zbN}
	U_0^{M,N} \xrightarrow[N \to \infty]{{\rm a.s.}} U_0^M
\end{equation}
and
\begin{equation}\label{zbM}
\lim_{M\to+\infty}U_0^M =U_0=V.
\end{equation}
\end{Thm}
\begin{proof}
The convergence \eqref{zbN} can be proved in a similar way as Theorem 3.2 of \cite{proof}.
Indeed, our model fulfills all of the assumptions stated in \cite{proof} and in our modified algorithm we do not change the way of approximating $\alpha_j^M$.

From \eqref{U} if follows that to prove \eqref{zbM} it sufficient to show that for all $0 \leq j \leq L$ we have:
\begin{equation}\label{T1}
	\lim\limits_{M \to \infty} \E [ Z_{\tau_j^M} | \F_{t_j}] = \E [Z_{\tau_j} | F_j].
\end{equation}

We will prove this fact using backward mathematical induction.
Equation (\ref{T1}) holds for $L$ since $Z_{\tau_L^M} = Z_{\tau_L} = Z_{T}$. Assume that (\ref{T1}) holds for $j+1$. To prove that it holds for $j$,
observe, that
\begin{equation}
	Z_{\tau_j^M} = Z_{t_j}\I_{\{ Z_{t_j}\Ij \geq [\alpha_j^M \cdot \phi^M(S_{t_j})]\Ij\} } + 	Z_{\tau_{j+1}^M}\I_{\{ Z_{t_j}\Ij < [\alpha_j^M \cdot \phi^M(S_{t_j})]\Ij\} }.
\end{equation}
Furthermore,
\begin{align}
	&\E [Z_{\tau_j^M} - Z_{\tau_j} | \F_{t_j}] =  	\E [Z_{\tau_{j+1}^M} - Z_{\tau_{j+1}} | \F_{t_j}]\I_{ \{Z_{t_j}\Ij < [\alpha_j^M \cdot \phi^M(S_{t_j})]\Ij\} } \\
	&\quad + (Z_{t_j} - \E[ Z_{\tau_{j+1}} |\F_{t_j}])\left( \I_{\{ Z_{t_j}\Ij \geq [\alpha_j^M \cdot \phi^M(S_{t_j})]\Ij\} } - \I_{\{ Z_{t_j}\Ij \geq \E[Z_{\tau_{j+1}} |\F_{t_j}]\Ij\} } \right). \nonumber
\end{align}

By induction assumption, $\E [Z_{\tau_{j+1}^M} - Z_{\tau_{j+1}} | \F_{t_j}]$ converges to 0 as $M$ goes to infinity. Let us define $B_j^M$ as the remaining part on the right hand side of above equality, that is,
\begin{equation}
	B_j^M = (Z_{t_j} - \E[ Z_{\tau_{j+1}} |\F_{t_j}])\left( \I_{ \{Z_{t_j}\Ij \geq [\alpha_j^M \cdot \phi^M(S_{t_j})]\Ij\} } - \I_{\{ Z_{t_j}\Ij \geq \E[Z_{\tau_{j+1}} |\F_{t_j}]\Ij\} } \right).
\end{equation}

We have
\begin{align*}
	|B_j^M| &= |Z_{t_j} - \E[ Z_{\tau_{j+1}} |\F_{t_j}]| |\I_{\{ Z_{t_j}\Ij \geq [\alpha_j^M \cdot \phi^M(S_{t_j})]\Ij\} } - \I_{\{ Z_{t_j}\Ij \geq \E[Z_{\tau_{j+1}} |\F_{t_j}]\Ij\} }| \\
	\nonumber &= |Z_{t_j} - \E[ Z_{\tau_{j+1}} |\F_{t_j}]| |\I_{\{ \E[Z_{\tau_{j+1}} |\F_{t_j}]\Ij > Z_{t_j}\Ij \geq [\alpha_j^M \cdot \phi^M(S_{t_j})]\Ij\} } \\
	\nonumber& - \I_{\{ [\alpha_j^M \cdot \phi^M(S_{t_j})]\Ij > Z_{t_j}\Ij \geq \E[Z_{\tau_{j+1}} |\F_{t_j}]\Ij \}}| \\
	\nonumber &\leq |Z_{t_j} - \E[ Z_{\tau_{j+1}} |\F_{t_j}]|\I_{ \{| Z_{t_j} - \E[ Z_{\tau_{j+1}} |\F_{t_j}] |\Ij \leq |\alpha_j^M \cdot \phi^M(S_{t_j}) - \E[ Z_{\tau_{j+1}} |\F_{t_j}] |\Ij \}} \\
	\nonumber &\leq |Z_{t_j} - \E[ Z_{\tau_{j+1}} |\F_{t_j}]|\I_{\{ | Z_{t_j} - \E[ Z_{\tau_{j+1}} |\F_{t_j}] | \leq |\alpha_j^M \cdot \phi^M(S_{t_j}) - \E[ Z_{\tau_{j+1}} |\F_{t_j}] | \}} \\
	\nonumber &\leq  |\alpha_j^M \cdot \phi^M(S_{t_j}) - \E[ Z_{\tau_{j+1}} |\F_{t_j}]|.
\end{align*}
By arguments similar to the one given in the proof of Theorem 3.1 of \cite{proof} we can conclude that $|\alpha_j^M \cdot \phi^M(S_{t_j}) - \E[ Z_{\tau_{j+1}} |\F_{t_j}]|$ converges to $0$ as $M$ tends to infinity.
Therefore $B_j^M$ converges to $0$ as $M$ goes to infinity, which completes the proof.
\end{proof}

\section{Numerical analysis}\label{sec:num}

In a general setup, $\theta$ is a stopping time with respect to $\{\mathcal{F}_t\}_{\{t\geq 0\}}$ related to an event observable by the market participants. We allow it to be both dependent and independent of the underlying asset price. However, we will focus particularly on the case in which $\theta$ is the moment of a drawdown exceeding a fixed level. More precisely, in this paper, we focus on a first time when the underlying asset price falls below its historical maximum value by a fixed percentage threshold, that is,
\begin{equation}\label{theta}
	\theta = \inf \{ t \geq 0: 1 - S_t\overline{S}_t^{-1} \geq C \},
\end{equation}
where $C$ is the drawdown size and
\begin{equation}\label{s_max}
	\overline{S}_t = \overline{s} \vee \max_{0 \leq s \leq t} S_t
\end{equation}
is the running maximum of underlying asset price.
Above $\overline{s}$ is the historical maximum of the underlying asset price until the issued date of the option. 
Additionally, we make a following assumption: 
$$\text{\underline{The jumps in the underlying asset have the exponential distribution with mean $\rho^{-1}$.}}$$

Our goal is to find a fair price for such a contract for a chosen set of market parameters. Additionally, we want to examine how changing the parameters affects the final price.

At first, we want to compare the prices of the capped option for underlying asset described by the geometric Browian motion (GBM) and the geometric L\'evy process. We choose the following set of parameters:
$S_0 = 100,\ \overline{s} = 105,\ K = 110,\ T = 1,\ r = 0.1,\ \rho = 0.5$. For GBM we set $\sigma = 0.4$ and for the L\'evy process we set $\sigma = 0.5,\ \lambda = 0.0675$. Such a choice of parameters leads to $\mu = 0.02$ in both cases. Then, for both models we choose identical $2000$ time steps for each trajectory, $5000$ trajectories for each price calculation, and the first $5$ functions from the Laguerre basis to approximate the conditional expected value of holding the option. Additionally, we repeat the calculations for different sizes of the drawdown level $C$ and perform each calculation 200 times to draw the boxplots of the results. The prices are shown in Figures \ref{BS} and \ref{K110}. In both figures, it can be seen that with $C$ closer to 1, the prices converge to the price of the standard American option. One can also observe that the prices are higher for the L\'evy market. Indeed, keeping the parameter $\mu$ the same in both cases results in a higher volatility of the underlying asset prices described by the geometric L\'evy process. As a consequence, the prices of derivatives written on this underlying increase as well. The convergence to the stardard American option price with increasing $C$ is also visible in Figure \ref{K90}, which contains results for the L\'evy market with the same parameters as before but with $K=90$. There one can also see that for $C = 0.1$ the price is close to 0, but still positive. This behaviour is only possible due to the negative jumps of the underlying asset price process.

All of the aforementioned figures show a sigmoidal shape of the curves that shows the dependence of the price on the level of the threshold $C$. In each chart, the price of the option converges to the immediate payoff at time $t=0$ with $C \rightarrow 0$ and to the price of the corresponding vanilla American option with $C \rightarrow 1$. Additionally, for all sets of parameters, the inflection point of the price curve seems to be somewhere between $C = 0.2$ and $C = 0.3$. The shape is best reflected in Figure \ref{K90}.

As a next step, we analyse the evolution of option prices with increasing price of the underlying at time 0, see Figure \ref{smooth_fit}. Here, we choose the following set of parameters: $C = 0.2,\ K = 110,\ \sigma = 0.2,\ \lambda = 0.015,\ \rho = 0.5,\ T = 1$. Additionally, we choose the set of initial prices as $S = \{ 90, 91, \ldots, 130 \}$ and set $\overline{s} = s$ for each $s \in S$. Figure \ref{smooth_fit} suggests that we can expect a smooth fit of the payoff function and the price curve of the options capped by a drawdown event.

Finally, we decided to analyse the dependence of the option price on $r$ and $\sigma$ parameters. We take most of the parameters as in the analysis presented in Figure \ref{smooth_fit} but set $S_0 = 100,\ \overline{S}_0 = 105$ and choose $r$ ranging from $0.01$ to $0.2$ and $\sigma$ ranging from $0.02$ to $0.4$. Figure \ref{sensis} gives the capped option prices for the chosen ranges of $r$ and $\sigma$. Note that the behaviour of the prices is similar to the one of standard American and European put option prices: increasing the volatility has a positive, and increasing the interest rate a negative impact on the price.

In Table \ref{table1} we present the average computational time for generating one boxplot (created with results of 200 simulations) for each of Figures \ref{BS}-\ref{smooth_fit}. We did our calculations in MATLAB on a computer with AMD Ryzen 7 5800H CPU and 16GB of RAM memory. For efficient calculation of equation (\ref{matrix}), we used built-in \textit{linsolve} function.

\begin{table}
	\begin{tabular}{ ||c||c|c|c|c||  } 
		\hline
		Figure number & 1 & 2  & 3 & 4 \\ 
		\hline
		Avg. time [s] & 559 & 456  & 549 & 522 \\ 
		\hline
	\end{tabular}
\vspace{3mm}
\caption{Average time of generating one boxplot generated with results from 200 independent simulations (in seconds).}
\label{table1}
\end{table}

	\begin{figure}
	\includegraphics[width=\textwidth]{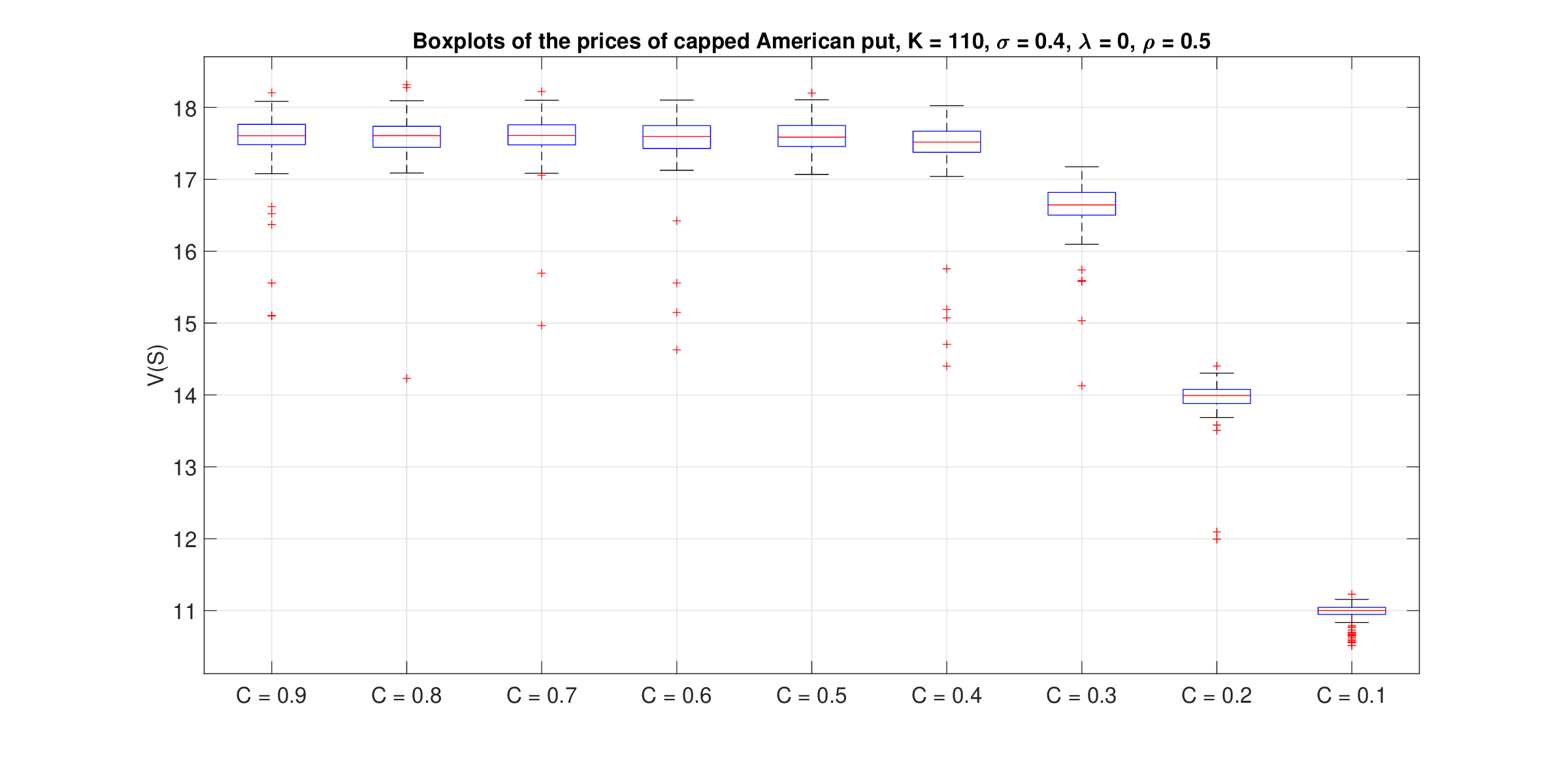}
	\caption{Boxplot of 100 prices of the capped American put option with underlying modelled by the geometric Brownian motion.}
	\label{BS}
\end{figure}

\begin{figure}
	\includegraphics[width=\textwidth]{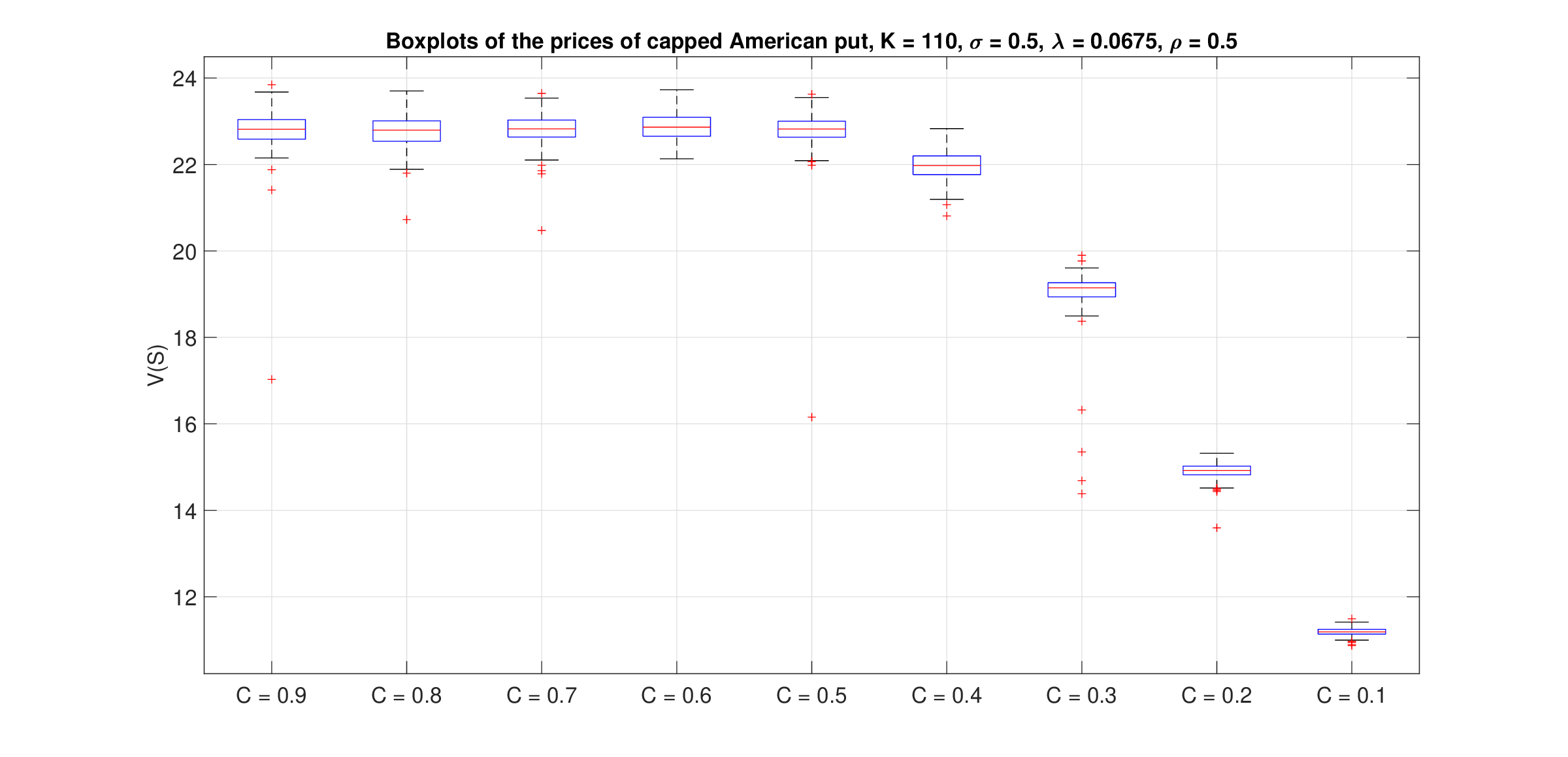}
	\caption{Boxplot of 100 prices of the capped American put option with underlying modelled by the geometric L\'evy process.}
	\label{K110}
\end{figure}

\begin{figure}
	\includegraphics[width=\textwidth]{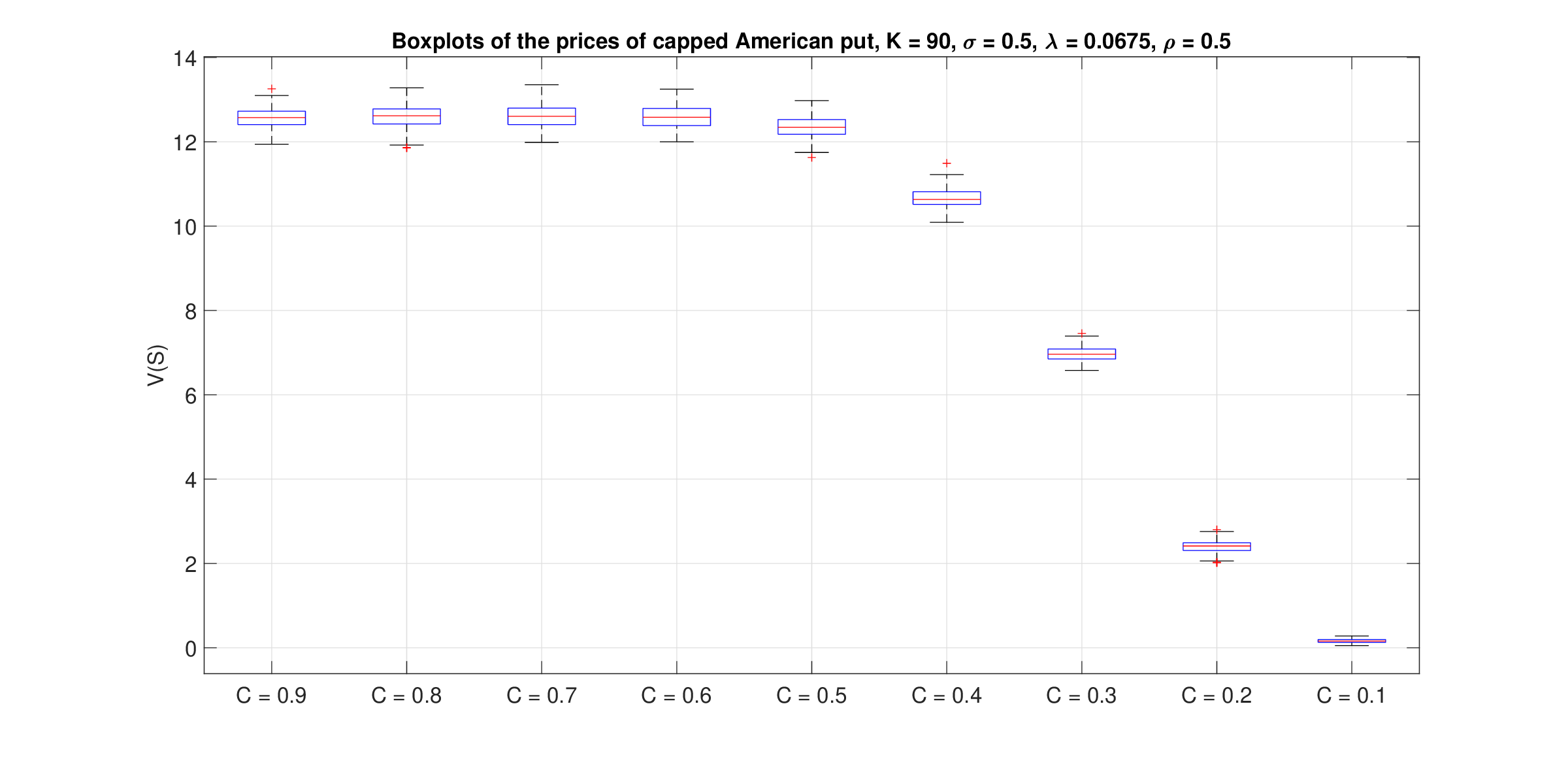}
	\caption{Boxplot of 100 prices of the capped American put option with underlying modelled by the geometric L\'evy process.}
	\label{K90}
\end{figure}

\begin{figure}
	\includegraphics[width=\textwidth]{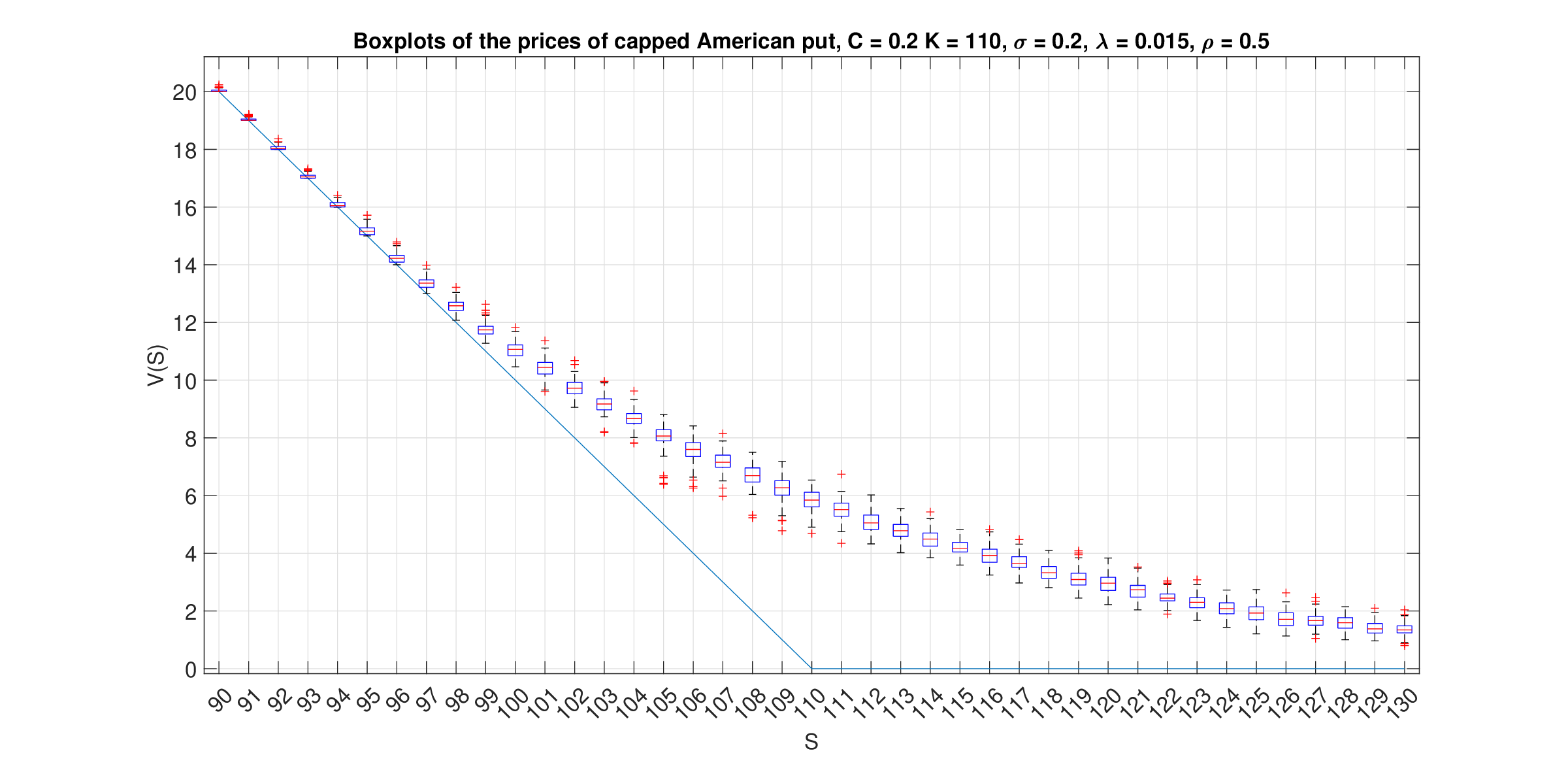}
	\caption{Smooth fit of the capped American put option prices and the payout function.}
	\label{smooth_fit}
\end{figure}

\begin{figure}
	\includegraphics[width=\textwidth]{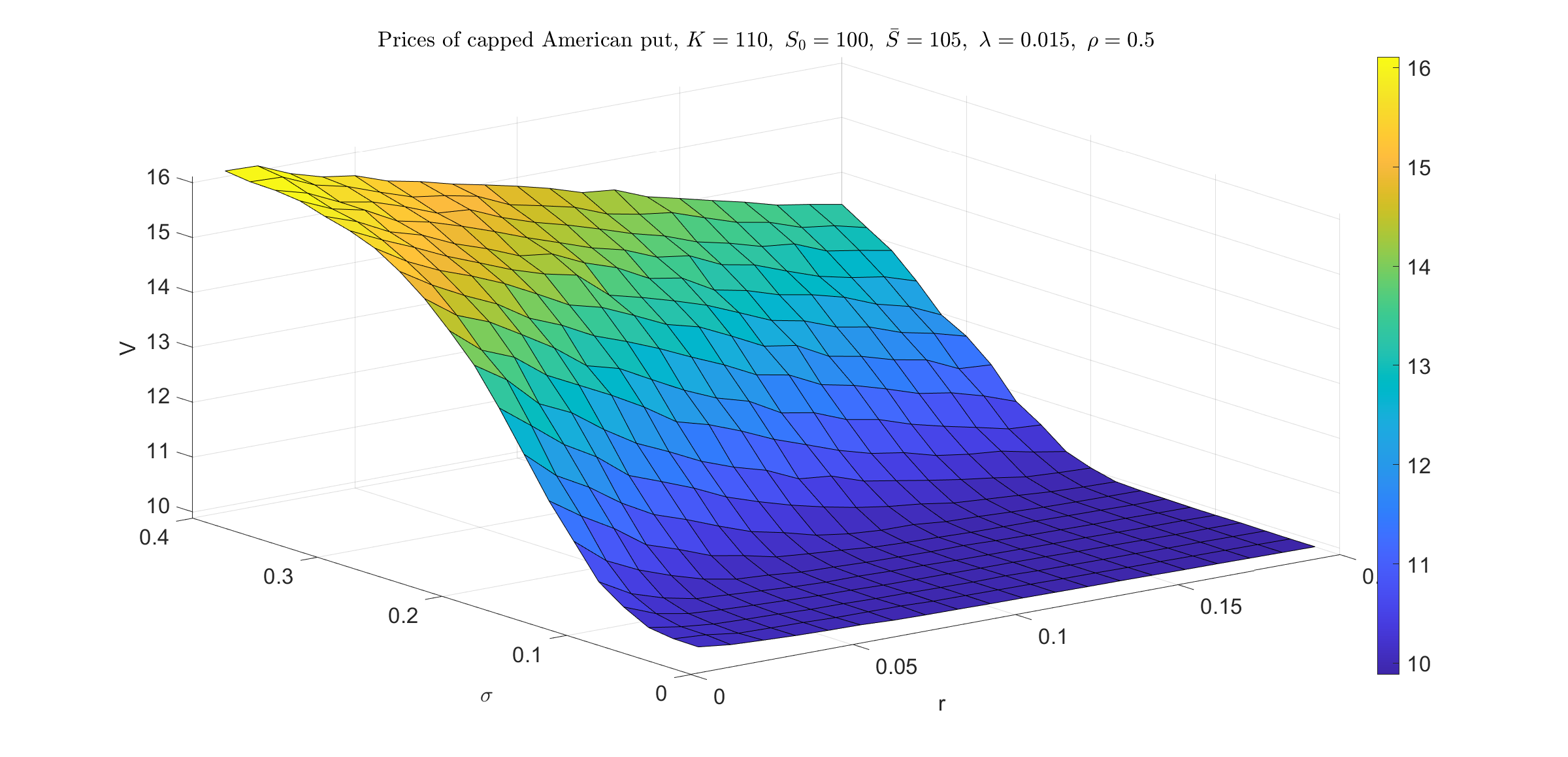}
	\caption{Price surface for the capped American put option for different values of parameters $\sigma$ and $\rho$.}
	\label{sensis}
\end{figure}

\newpage

\section*{Declarations of Interest}
The authors report no conflicts of interest. The authors alone are responsible for the content and writing of the paper.


\begin{thebibliography}{10}
\expandafter\ifx\csname url\endcsname\relax
  \def\url#1{\texttt{#1}}\fi
\expandafter\ifx\csname urlprefix\endcsname\relax\def\urlprefix{URL }\fi
\expandafter\ifx\csname href\endcsname\relax
  \def\href#1#2{#2} \def\path#1{#1}\fi


\bibitem{Aase2010}
Aase, K.K. (2010)
The perpetual american put option for jump-diffusions.
In Endre Bj{\o}rndal, M. Bj{\o}rndal, P.~M. Pardalos, and
  M. R{\"o}nnqvist, editors, {\em Energy, Natural Resources and
  Environmental Economics}, 493--507. Springer Berlin Heidelberg.

\bibitem{AliliKyprianou}
Alili, A. and Kyprianou, A.  (2005)
Some remarks on first passage of {L}\'evy process, the {A}merican put
  and pasting principles.
{\it Ann. Appl. Probab.} {\bf 15}, 2062--2080.

\bibitem{AsmussenAvramPistorius}
Asmussen, S., Avram, F. and Pistorius, M. (2004)
Russian and {A}merican put options under exponential phase-type
  {L}\'evy models.
{\it Stoch. Process. Appl.} {\bf 109}, 79--111.

\bibitem{Florin}
Avram, F., Chan, T.  and Usabel, M. (2002)
On the valuation of constant barrier options under
spectrally one-sided exponential L\'evy models and
Carr’s approximation for American puts.
{\it Stoch. Process. Appl.} {\bf 100}, 75–107.

\bibitem{BieleckiRutkowski}
Bielecki, T.R. and Rutkowski, M. (2014)
{\it Credit Risk: Modeling, Valuation and Hedging}.
Springer.

\bibitem{Erik2}
Baurdoux, E. (2007)
{\it Fluctuation Theory and Stochastic Games for Spectrally Negative
  {L}\'evy Processes}.
Phd Thesis, Utrecht University.

 \bibitem{Erik1}
Baurdoux, E. and Kyprianou, A. (2008)
The {M}c{K}ean stochastic game driven by a spectrally negative
  {L}\'evy process.
{\it Electronic Journal of Probability} {\bf 8}, 173--197.

\bibitem{boyarchenkolevendorskii}
Boyarchenko, S. I. and Levendorskii, S. Z. (2002)
Perpetual {A}merican options under {L}\'evy processes.
{\it SIAM Journal of Control and Optimization} {\bf 40}, 1663--1696.

\bibitem{cap} Broadie, M. and Detemple, J. (1995) American Capped Call Options on Dividend-Paying Assets. {\it The Review of Financial Studies} {\bf 8(1)} : 161--91.

\bibitem{carr} Carr, P. (1998) Randomization and the American Put. {\it The Review of Financial Studies} {\bf 11(3)}.

\bibitem{CZH}
Carr, P., Zhang, H. and Hadjiliadis, O. (2011) Maximum drawdown insurance. {\it International
Journal of Theoretical and Applied Finance} {\bf 14(8)}, 1--36.

\bibitem{B24}
Carr, P., Madan, D., Geman, H. and Yor, M. (2002)
The fine structure of asset returns, an empirical investigation.
{\it J. Business}, 75(2), 305--332.

\bibitem{carr_wu} Carr, P. and Wu, L. (2003) The Finite Moment Log Stable Process and Option Pricing. {\it The Journal of Finance} {\bf 58}, 753--777. 

\bibitem{proof} Clement, E. and Lamberton, D. and Protter, P. (2002) An analysis of a least squares regression method for American option pricing. {\it Finance and Stochastics} {\bf 6(4)}, 449--471.

\bibitem{Cont}
Cont R. and Tankov, P. (2003)
{\it Financial modelling with jump processes}.
Chapman and Hall, CRC Press.

\bibitem{B42}
Eberlein E. and Keller, U. (1995)
Hyperbolic distributions in finance.
{\it Bernoulli} {\bf 1}, 281--299.

\bibitem{time_constraints} Egloff, D., Farkas, W. and Leippold, F. M. (2009) American Options with Stopping Time Constraints.
{\it Applied Mathematical Finance} {\bf 16(3)}, 287--305. 

\bibitem{gapeev1} Gapeev, P. and Al Motairi, H (2018) Perpetual American defaultable options in models with random dividends and partial information. {\it Risks} {\bf 6}, 127.

\bibitem{algo} Gapeev, P., Li, L. and Wu, Z. (2020) Perpetual American Cancellable Standard Options in Models with Last Passage Times. {\it Algorithms} {\bf 14}, 3.

\bibitem{1st_pass} Gapeev, P. and Al Motairi, H. (2022) Discounted optimal stopping problems in first-passage time models with random thresholds. {\it J. Appl. Probab.} {\bf 59}, 714--33.

\bibitem{GZ}
Grossman, S. J. and Zhou, Z. (1993) Optimal investment strategies for controlling drawdowns.
{\it Mathematical Finance} {\bf 3(3)}, 241--276.

\bibitem{kifer} Kifer, Y. (200) Game options. {\it Finance Stoch.} {\bf 4}, 443--463.

\bibitem{B68}
Koponen, I. (1995)
Analytic approach to the problem of convergence of truncated {L}\'evy
  flights towards the {G}aussian stochastic process.
{\it Phys. Rev. E.} {\bf 52}, 1197--1199.

\bibitem{lsmc} Longstaff, F. A.and Schwartz, E. S. (2001) Valuing American Options by Simulation: A Simple Least-Squares Approach. {\it The Review of Financial Studies} {\bf 14(1)}.

\bibitem{B80}
Madan, D. and Seneta, E. (1990)
The variance gamma model for share market returns.
{\it J. Business} {\bf 63}, 511--524.

\bibitem{MA}
Magdon-Ismail, M. and Atiya, A. (2004) Maximum drawdown. {\it Risk} {\bf 17(10)}, 99--102.

\bibitem{Merton}
Merton, R. (1976)
Option pricing when the underlying stock returns are discontinuous.
{\it Journal of Financial Economics} {\bf 3}, 125--144.

\bibitem{Mordecki}
Mordecki, E. (2002)
Optimal stopping and perpetual options for {L}\'evy processes.
{\it Finance Stoch.} {\bf 6(4)}, 473--493.

\bibitem{meyer} Meyer, G. H.  (2016) A PDE View of Games Options. {\it Research Paper Series} {\bf 369}, Quantitative Finance Research Centre .


\bibitem{B10}
Nielssen, O. B. (1998)
The {M}c{K}ean stochastic game driven by a spectrally negative
  {L}\'evy process.
{\it Finance Stoch.} {\bf 1}, 41--68.

\bibitem{ott} C. Ott, C. (2013) Optimal Stopping Problems for the Maximum Process. PhD thesis, University of Bath.

\bibitem{ZP&JT2}
Palmowski, Z. and Tumilewicz, J. (2018)
Pricing insurance drawdowns-type contracts with underlying L\'evy assets.
{\it Insurance: Mathematics and Economics} {\bf 79}, 1--14.

\bibitem{ZP&JT}
Palmowski, Z. and Tumilewicz, J. (2020)
Fair valuation of L\'evy-type drawdown-drawup contracts with general insured and penalty functions.
{\it Applied Mathematics and Optimization} {\bf 81}, 301--347.


\bibitem{drawdownup2}
Pospisil, L., Vecer, J. and Hadjiliadis, O. (2009) Formulas for stopped diffusion processes with stopping times based on drawdowns and drawups,
{\it Stoch. Process. Appl.} {\bf 119(8)}, 2563--2578.

\bibitem{PV}
Pospisil, L. and Vecer, J. (2010) Portfolio sensitivities to the changes in the maximum and
the maximum drawdown. {\it Quantitative Finance} {\bf 10(6)}, 617--627.

\bibitem{Schoutens}
Schoutens, W. (2003)
{\it {L}\'evy Processes in Finance: Pricing Financial Derivatives}.
Wiley.

\bibitem{Sorn}
Sornette, D. (2003) {\it Why Stock Markets Crash: Critical Events in Complex Financial Systems.}
Princeton University Press.

\bibitem{stentoft} Stentoft, L. (2004) Convergence of the Least Squares Monte Carlo Approach to American. Option Valuation. {\it Management Science}.

\bibitem{trabelsi} Trabelsi, F. (2011) Asymptotic Behavior of Random Maturity American Options. {\it IAENG International Journal of Applied Mathematics}, {\bf 41}, 2.

\bibitem{Vec1}
Vecer, J. (2006) Maximum drawdown and directional trading. {\it Risk} {\bf 19(12)}, 88--92.

\bibitem{Vec2}
Vecer, J. (2007) Preventing portfolio losses by hedging maximum drawdown. {\it Wilmott} {\bf 5(4)},
1--8.

\bibitem{random_put} Wu, Z. and Li, L. (2022) The American Put Option with a Random Time Horizon. arXiv:2211.13918.

\bibitem{russian} Wu, Z. and Li, L. (2022) The Russian Option with A Random Time Horizon. arXiv:2211.13917v1.

\bibitem{drawdownup1}
Zhang, H. and Hadjiliadis, 0. (2009) Formulas for the Laplace transform of stopping times
based on drawdowns and drawups. http://avix.org/pdf/0911.1575.

\bibitem{olympia}
Zhang, H., Leung, T. and Hadjiliadis, O. (2013) Stochastic modeling and fair valuation of drawdown insurance. {\it Insurance: Mathematics and Economics} {\bf 53}, 840--850.

\bibitem{Zaevski2} Zaevski, T. (2020) Discounted perpetual game call options. {\it Chaos, Solitons \& Fractals} {\bf 131}.

\bibitem{Zaevski2b} Zaevski, T. (2020) Discounted perpetual game put options. {\it Chaos, Solitons \& Fractals} {\bf 137}, 109858.

\bibitem{Zaevski} Zaevski, T. (2022) Pricing discounted American capped options. {\it Chaos, Solitons \& Fractals} {\bf 156}. 

\bibitem{Mladen}
Zaevski, T., Kounchev, O. and Savov, M. (2019)
Two frameworks for pricing defaultable derivatives.
{\it Chaos, Solitons \& Fractals} {\bf 123}, 309--319.




\end{thebibliography}
\end{document}